\documentclass[11pt]{amsart}

\usepackage[centertags]{amsmath}
\usepackage{amsthm,amsfonts}
\usepackage{amssymb}
\usepackage{epsfig}
\usepackage{amssymb,latexsym}
\usepackage{caption}
\usepackage{indentfirst}

\usepackage{mathrsfs}
\usepackage{amsthm}
\usepackage{amsmath}
\usepackage{amssymb}
\usepackage{amsfonts}
\usepackage{amsbsy}
\usepackage{tikz}
\usepackage{pdflscape}
\usetikzlibrary{matrix,fit,shapes.misc,positioning}
\tikzset{%
	highlight/.style={rectangle,rounded corners,fill=red!15,draw,fill opacity=0.3,thick,inner sep=0pt}
}
\tikzset{%
	highlight1/.style={rectangle,rounded corners,fill=blue!15,draw,fill opacity=0.3,thick,inner sep=0pt}
}

\theoremstyle{plain}
\newtheorem{thm}{Theorem}[section]

\newtheorem{lem}[thm]{Lemma}
\newtheorem{prop}[thm]{Proposition}

\theoremstyle{definition}

\newtheorem{rem}[thm]{Remark}

\numberwithin{equation}{section}

\newcommand{\F}{{\mathbb F}}

\usepackage[bottom]{footmisc}

\begin{document}
	
	\title{Linear Complementary Pair Of Group Codes over Finite Chain Rings}
	\author{Cem G\"uner\.{I}, Edgar Mart\'inez-Moro, Selcen Say\i c\i}
	\thanks{C. G\"{u}neri and S. Say\i c\i\ are with Sabanc{\i} University, Faculty of Engineering and Natural Sciences,  \.{I}stanbul, Turkey. Email: guneri@sabanciuniv.edu, selcensayici@sabanciuniv.edu}
	\thanks{E. Martinez-Moro is with University of Valladolid, Institute of Mathematics, Castilla, Spain. Email: edgar.martinez@uva.es}

	\abstract
	Linear complementary dual (LCD) codes and linear complementary pair (LCP) of codes over finite fields have been intensively studied recently due to their applications in cryptography, in the context of side channel and fault injection attacks. The security parameter for an LCP of codes $(C,D)$ is defined as the minimum of the minimum distances $d(C)$ and $d(D^\bot)$. It has been recently shown that if $C$ and $D$ are both 2-sided group codes over a finite field, then $C$ and $D^\bot$ are permutation equivalent. Hence the security parameter for an LCP of  2-sided group codes $(C,D)$  is simply $d(C)$. We extend this result to 2-sided group codes over finite chain rings.

	\vspace{0.3cm}
	
	\endabstract
	\maketitle

	\section{Introduction}\label{intro}

	A pair of linear codes $(C,D)$ of length $n$ over a finite field $\F_q$ is called a linear complementary pair (LCP) of codes if $C \cap D = \{0\}$  and $C + D = \F_{q}^{n}$ (i.e. $C \oplus D = \F_{q}^{n}$). In the case $D=C^\bot$, $C$ is referred as a linear complementary dual (LCD) code. 
	
	LCD codes were introduced by Massey \cite{M} in 1992. There has been a revived interest in LCD and LCP of codes due to their application in protection against side channel and fault injection attacks (\cite{BDGNN, BCCGM2014}). In this context, the security parameter of an LCP $(C,D)$ is defined to be $\min\{d(C),d(D^\bot)\}$, where $d(C)$ stands for the minimum distance of the code $C$. In the LCD case, this parameter is simply $d(C)$, since $D^\bot=C$. 
	
	Carlet et al. (\cite{CGOOS}) showed that if $(C,D)$ is LCP, where $C$ and $D$ are both cyclic codes over a finite field $\F_q$, then $C$ is equivalent to $D^\bot$. They showed that the same result holds if $C$ and $D$ are 2D cyclic codes, under the assumption that the length of the codes is relatively prime to the characteristic of the finite field (semisimple case). Cyclic and 2D cyclic codes are special abelian codes, which are defined as ideals of the group algebra $\F_q[G]$ for a finite abelian group $G$.  In the case $\gcd(q,|G|)=1$, G\"{u}neri et al. extended this result  to LCP of abelian codes in $\F_q[G]$ (\cite{GOS}). If $G$ is any finite group (not necessarily abelian), a right ideal of $\F_q[G]$ is called a group code. In \cite{BCW}, Borello et al. obtained the most general statement for any finite group (also without a restriction on the order of the group) by showing that if $(C,D)$ is LCP of 2-sided group codes (ideals) in $\F_q[G]$, then $C$ and $D^\bot$ are permutation equivalent. Note in particular that this implies $d(C)=d(D^\bot)$. Hence, there is an LCP of 2-sided group codes over finite fields which has as good a security parameter as the 2-sided group code with the best minimum distance.
	
Although LCD and LCP of codes have been extensively studied over finite fields, the literature on codes over rings in this context is very limited. LCD codes over rings, including chain rings, have been addressed recently in  \cite{Rama,LL,LW}.  Here, we study LCP of codes over rings for the first time. Our main contribution is the extension of the result in \cite{BCW} to finite chain rings. Namely, we prove that for an LCP of 2-sided group codes $(C,D)$ in $R[G]$, where $R$ is a finite chain ring and $G$ is any finite group, $C$ and $D^\bot$ are equivalent codes (Theorem \ref{main result-3}).

	\section{Background}\label{generators}
Unless otherwise stated $R$ denotes a chain ring, which is a finite commutative ring with identity whose lattice of ideals forms a chain. Note that a finite field is a special chain ring.  It is clear that $R$ is a local ring and it is well-known that $R$ is a principal ideal ring. Let $\gamma$ be a generator of the maximal ideal and let the ideals of $R$ be 
$$R=R\gamma^0\supset R\gamma \supset \cdots \supset R\gamma^{v-1} \supset R\gamma^v=\{0\}.$$
The number $v$ with $\gamma^v=0$ is called the nilpotency index of $\gamma$. Note that since  $R$ is a commutative ring, $ R\gamma^{i}   = \gamma^{i}R          $ for all $i$.

It is clear that $R/R\gamma$ is a finite field, which we will denote by $\F_q$. The natural projection map $\varphi: R \rightarrow \F_q$ takes a ring element to its coset modulo $R\gamma$. This map is a surjective ring homomorphism and it extends to $R^n$ and takes values in $\F_q^n$ via
	\begin{equation} \label{the map}
	 (r_i) \longmapsto \varphi(r_i),
	\end{equation}
where $(r_i)$ denotes an $n$-tuple over $R$. We will denote the extended map by $\varphi$ as well, which is a surjective $R$-module homomorphism. The kernel of this map is the set of all $n$ tuples whose coordinates are multiples of $\gamma$ (i.e. $(\gamma R)^n$). We will also denote this set with $\gamma R^n$. Observe that $\varphi$ maps an $R$-submodule of $R^n$ to an $\F_q$-subspace of $\F_q^n$. An $R$-submodule of $R^n$ is called a linear code over $R$. Hence, $\varphi$ maps a linear code over $R$ to a linear code over $\F_q$.

Now, let $G$ be a finite group and denote by $R[G]$ the group ring of $G$ over $R$. Hence the elements of $R[G]$ are of the form $\sum_{g \in G} \alpha_{g} g$, where $\alpha_{g} \in R$ and $\alpha_g$ is nonzero for finitely many $g\in G$.  If $G$ has order $n$, then it is clear that $R[G]$ and $R^n$ are isomorphic as $R$-modules, where an element  $\sum_{g \in G} \alpha_{g} g \in R[G]$ is identified with the $n$ tuple $(\alpha_g)$. We will use this identification throughout the text. The group rings will be specifically used when we have results which are valid for group codes over $R$. A right ideal of $R[G]$ is called a group code over $R$ (see \cite{BCW} for group codes over finite fields). Our main result (Theorem \ref{main result-3}) holds for 2-sided ideals in $R[G]$. Therefore, unless otherwise stated, ideals will be 2-sided throughout and and they will be referred to as group codes. If $G$ is abelian, then a group code (ideal) in $R[G]$  is called an abelian code over $R$. 

\begin{rem}\label{isom of group rings}
If $G$ and $G'$ are finite multiplicative groups which are isomorphic via a map $\psi$, and if $R$ is any ring, then it is easy to see that $\psi$ extends to a ring isomorphism
		\[\begin{array}{cccc}
		\psi : & R[G] & \longrightarrow & R[G']\\
		& \displaystyle{\sum_{g\in G} r_g g} & \longmapsto & \displaystyle{\sum_{g\in G}r_g \psi(g)}
		\end{array}\]
Hence such a map takes a group code in $R[G]$ to a group code in $R[G']$. If $G=G'$, we can consider an automorphism of $G$ as a permutation on $G$. Note that an arbitrary permutation of $G$ does not necessarily preserve the ideal structure in $R[G]$ but those which are automorphisms do. 
\end{rem}
	\vspace{0.5cm}
A pair of linear codes $(C,D)$ in $R^n$ is called a linear complementary pair (LCP) of codes if $C\oplus D=R^n$. When $D=C^\bot$, $C$ is said to be a linear complementary dual (LCD) code over $R$. The dual in this article will always be considered with respect to the Euclidean inner product on $R^n$. It is easy to see that the dual of a group code in $R[G]$ is also a group code. 

For a finite field $\F$ and an arbitrary finite group $G$, consider LCP of (2-sided) group codes $(C,D)$ in $\F[G]$. Borello et al. showed in \cite{BCW} that $C$ is permutation equivalent to $D^\bot$. The permutation yielding the equivalence, which we will later denote by $\tau$, is the inversion automorphism that takes $g$ to $g^{-1}$, for all $g\in G$. We will extend this equivalence result to LCP of group codes over finite chain rings. 


\section{LCP of Group Codes over Chain Rings}\label{LCP chain ring}

We use the notation and notions introduced in Section \ref{generators}. In particular, all codes over $R$ are linear (i.e. $R$-module) and a group code in $R[G]$ is a 2-sided ideal. 

\begin{lem}\label{lem1} If $(C,D)$ is LCP of codes in $R^n$, then both $C$ and $D$ are free modules (codes).
\end{lem}
	
\begin{proof}
Note that by definition (being direct summands of the free module $R^n$), both $C$ and $D$ are projective modules over $R$. A chain ring is local and by \cite[Theorem 2]{K}, a projective module over a local ring is free.
\end{proof}

\begin{prop}\label{lem2}
(i) If $(C,D)$ is LCP of codes in $R^n$, then $(\varphi(C),\varphi(D))$ is LCP of codes in $\F_q^n$. 

(ii) If $(C,D)$ is LCP of group codes in $R[G]$, then $(\varphi(C),\varphi(D))$ is LCP of group codes in $\F_q[G]$.
\end{prop}

\begin{proof}
(i) Let $x\in \F_q^n$. Since $R^n$ is the direct sum of $C, D$, and  $\varphi$ is surjective, there exist $c\in C, d\in D$ such that $x=\varphi(c)+\varphi(d)$. Hence, $\F_q^n$ is the sum of $\varphi(C)$ and $\varphi(D)$.

Let $x$ be in the intersection $\varphi(C)\cap \varphi(D)$. Then $x=\varphi(c)=\varphi(d)$, for some $c\in C, d\in D$. This gives $\varphi (c-d)=0$, and hence $(c-d)\in \gamma R^n$. Therefore, $\gamma^{v-1}(c-d)=0$. Set 
$$z:=\gamma^{v-1}c=\gamma^{v-1}d.$$
Note that $z$ is in $C\cap D$, which is by assumption trivial. So, $z=\gamma^{v-1}c=0$, which yields $c\in \gamma R^n$. Hence, $x=\varphi(c)=0$ and $\varphi(C)\cap \varphi(D)=\{0\}$.

(ii) We need to show that a left ideal $C\subset R[G]$ is mapped to a left ideal $\varphi(C) \subset \F_q[G]$, since the rest follows by part (i). For this, it suffices to show that $\varphi(C)$ is closed under left multiplication by an arbitrary element $g'\in G$, since being closed under left multiplication by a general element in $\F_q[G]$ then follows by linearity. If $\sum_g c_g g \in C$, then
\begin{equation*}
g'\varphi\left( \sum_{g} c_g g \right) =  g' \sum_{g} \varphi(c_g) g =  \sum_{g} \varphi(c_g) g'g = \varphi \left(g'\left(\sum_g c_g g  \right) \right).
\end{equation*}
Since $C$ is a left ideal, $g'\sum_g c_g g \in C$. Hence, $\varphi(C)$ is a left ideal in $\F_q[G]$. The proof for right ideal property is identical. 
\end{proof}
	
	For an element $r\in R$ and $x \in R^n$, $rx$ denotes the scalar multiplication, where each coordinate of $x$ is multiplied by $r$. For a code $C$ in $R[G]$, we set $rC:=\{rc: c\in C\}$. We define
	$$(C:r):=\{x\in R^n: \ rx\in C\},$$
which is a linear code in $R^n$. It is clear that
	$$C=(C:\gamma^0)\subseteq (C:\gamma)\subseteq \cdots \subseteq (C:\gamma^{v-1}),$$
	which implies 
	$$\varphi(C)=\varphi((C:\gamma^0))\subseteq \varphi((C:\gamma))\subseteq \cdots \subseteq \varphi((C:\gamma^{v-1})).$$
	We collect some facts which will be needed. Let us note that the dual code of $C \subset R^n$ (with respect to the Euclidean product) is defined as in codes over finite fields, and it is denoted by $C^\bot$.
	
	\begin{prop} (\cite[Theorem 3.10]{NS1}) \label{NS-1}
		Let $C$ be a code in $R^n$. Then,
		
		(i) $|C^\bot|=|R^n|/|C|$.
		
		(ii) $\varphi((C:\gamma^{v-1-i}))^\bot=\varphi((C^\bot:\gamma^i))$, for all $i$. 
	\end{prop}

	\begin{prop} (\cite[Proposition 3.13]{NS1}, \cite[Proposition 3.11 and Corollary 3.12]{NS2}) \label{NS-2}
		The following holds for a free code $C$ in $R^n$. 
		
		(i) $C^\bot$ is free.
		
		(ii) $\varphi(C)= \varphi((C:\gamma))= \cdots = \varphi((C:\gamma^{v-1}))$. 
		
		(iii) $C\cap \gamma^iR^n=\gamma^iC$, for all $i$.
		
		(iv) For $\tilde{C}:=C \setminus \gamma R^n=C\setminus \gamma C$, we have $C=\tilde{C}\cup \gamma \tilde{C} \cup \cdots \cup \gamma^{v-1} \tilde{C}\cup \{0\}$.
	\end{prop}
	
	We are ready to proceed with the steps of our proof.
	
	\begin{prop}\label{dual LCP}
		If $(C,D)$ is LCP of codes in $R^n$, then $(C^\bot,D^\bot)$ is also LCP.
	\end{prop}
	
	\begin{proof}
		Let $x$ be an element of $C^\bot \cap D^\bot$ and let $u=u_C+u_D$ be an arbitrary element in $R^n$, where $u_C\in C$ and $u_D\in D$. Then the Euclidean product of $x$ and $u$ is
		$$x\cdot (u_C+u_D)=x\cdot u_C + x\cdot u_D=0,$$
		since $x$ is orthogonal to both $C$ and $D$. So, $x=0$ since its inner product with any element in $R^n$ is 0. Therefore $C^\bot \cap D^\bot =\{0\}$. 
		
		For $c,c'\in C^\bot$ and $d,d'\in D^\bot$, if $c+d=c'+d'$ then $c-c'=d'-d \in C^\bot \cap D^\bot$. But this intersection is shown to be trivial, hence $c=c'$ and $d=d'$. Therefore the number of elements in $C^\bot +D^\bot = \{c'+d': \ c'\in C^\bot, d' \in D^\bot\}$ is $|C^\bot| |D^\bot|$. By Proposition \ref{NS-1}, 
		$$|C^\bot| |D^\bot| = \frac{|R^n|^2}{|C||D|}=|R^n|.$$
		Hence, $C^\bot + D^\bot =R^n$. The result follows since the two dual codes intersect only at 0.
	\end{proof}
	
	\begin{prop}\label{duals}
		(i) For a free code $C\subset R^n$, we have $\varphi(C)^\bot=\varphi(C^\bot)$. 
		
		(ii) If $(C,D)$ is LCP of group codes in $R[G]$, then $\varphi(C)$ and $\varphi(D^\bot)$ are equivalent codes.
	\end{prop}
	
	\begin{proof}
		(i) We have $\varphi(C)^\bot=\varphi((C^\bot:\gamma^{v-1}))$ by Proposition \ref{NS-1}. By Proposition \ref{NS-2} ((i) and (ii)), $\varphi((C^\bot:\gamma^{v-1}))=\varphi(C^\bot)$ for the free code $C^\bot$. Hence the result follows.
		
		(ii) By Proposition \ref{lem2}, $(\varphi(C), \varphi(D))$ is LCP of group codes in $\F_q[G]$. Then by \cite{BCW} (cf. Section \ref{generators}), $\varphi(C)$ and $\varphi(D)^\bot$ are equivalent group codes.  The result follows since $D$ is a free code and we have $\varphi(D)^\bot=\varphi(D^\bot)$ by part (i).
	\end{proof}
	
	\begin{rem} \label{proof-obs1}
		Note that for an LCP of group codes $(C,D)$ in $R[G]$, we have
		\begin{eqnarray*} \label{cardinalities}
		|D^\bot| & = & \frac{|R[G]|}{|D|} \ \ \mbox{(by Proposition \ref{NS-1}(i))} \nonumber \\
		&=& \frac{|C||D|}{|D|} \ \ \mbox{(since $C\oplus D=R[G]$)}  \\
		&=& |C|. \nonumber
		\end{eqnarray*}
		Let $\tau$ denote the permutation between $\varphi(C)$ and $\varphi(D^\bot)$ (\cite{BCW}). Then, 
$$\varphi(\tau(C))=\tau(\varphi(C))=\varphi(D^\bot).$$
For a free code over $R$, the minimum distance is equal to the minimum distance of its image under $\varphi$ (\cite[Corollary 4.3]{NS2}). A permutation clearly preserves the minimum distance. Hence, we have
		$$d(C)=d(\tau(C))=d(\varphi(\tau(C)))=d(\varphi(D^\bot))=d(D^\bot).$$
Our aim is to lift the equivalence $\tau$  between $\varphi(C)$ and $\varphi(D^\bot)$ to an equivalence between $C$ and $D^\bot$, whose cardinalities and minimum distances have been shown to be equal. 
	\end{rem}
	\vspace{0.5cm}
	From this point on, we consider an LCP of group codes $(C,D)$ in $R[G]$, since we will build up a proof for the main result (Theorem \ref{main result-3}) from the permutation equivalence between $\varphi(C)$ and $\varphi(D^\bot)$ (cf. Proposition \ref{duals}, Remark \ref{proof-obs1}). However, note that Proposition \ref{two representations} is true more generally (for free codes in $R^n$). 

If we restrict the map $\varphi : R[G] \rightarrow \F_q[G]$ to the (free) group codes $C$ and $D^\bot$, and use Proposition \ref{NS-2}(iii), we obtain the isomorphisms
	\begin{equation} \label{isoms}
	C/(C\cap \gamma R[G]) = C/\gamma C \simeq \varphi(C) \ \ \mbox{and} \ \ D^\bot/(D^\bot \cap \gamma R[G]) = D^\bot /\gamma D^\bot \simeq \varphi(D^\bot).
	\end{equation}
	Let $t:=|\varphi(C)|=|\varphi(D^\bot)|$ and set the elements of the cosets $C/\gamma C$ and $D^\bot/\gamma D^\bot$ as follows:
	\begin{eqnarray*}
		C/\gamma C & := & \left\{c_1+\gamma C=\gamma C, c_2+\gamma C, \ldots , c_t+\gamma C \right\},\\
		D^\bot/\gamma D^\bot & := & \left\{d_1+\gamma D^\bot=\gamma D^\bot, d_2+\gamma D^\bot, \ldots , d_t+\gamma D^\bot \right\}.
	\end{eqnarray*} 
	(i.e. $c_1=0=d_1$ in $R[G]$). Clearly, cosets partition the codes $C$ and $D^\bot$:
	\begin{equation}\label{disj unions}
	C=\underset{1\leq i \leq t} {{\dot{\bigcup}}} (c_i+\gamma C)  \ \ \mbox{and} \ \ D^\bot =\underset{1\leq i \leq t} {{\dot{\bigcup}}} (d_i+\gamma D^\bot)
	\end{equation}
	Note that $\varphi$ is constant on cosets, since a multiple of $\gamma$ is mapped to 0. Namely for all $i=1,\ldots ,t$, we have
	\[\begin{array} {cccccc}
	\varphi(c_i+\gamma c) & = & \varphi(c_i)+\varphi(\gamma c) & = & \varphi(c_i) & \mbox{for all $c\in C$}, \\
	\varphi(d_i+\gamma d) & = & \varphi(d_i)+\varphi(\gamma d) & = & \varphi(d_i) & \mbox{for all $d\in D^\bot$}.
	\end{array}\]
	Moreover $\varphi(c_i)\not= \varphi (c_j)$ (for $i\not= j$), since otherwise $c_i$ and $c_j$ would be in the same coset modulo $\gamma C$. The same holds for representatives of cosets of $D^\bot$ modulo $\gamma D^\bot$. Hence, we have 
	\begin{eqnarray}
	\varphi(C) & = & \left\{\varphi (c_1)=0, \varphi (c_2), \ldots , \varphi(c_t) \right\}, \nonumber \\
	\varphi(D^\bot) & = & \left\{\varphi (d_1)=0, \varphi (d_2), \ldots , \varphi(d_t) \right\}. \nonumber
	\end{eqnarray}
	Without loss of generality, we assume that the coset representatives are indexed so that the permutation $\tau$ between the equivalent codes $\varphi(C)$ and $\varphi(D^\bot)$ (cf. Remark \ref{proof-obs1}) satisfies 
	\begin{equation}\label{perm assump}
	\tau(\varphi(c_i))=\varphi(\tau(c_i))=\varphi (d_i), \ \ \mbox{for all $i=1,\ldots ,t$}.
	\end{equation}
	Note that this implies
	\begin{equation}\label{perm assump-2}
	\tau(c_i)-d_i \in \gamma R[G] \ \ \mbox{for all $i=1,\ldots ,t$}.
	\end{equation}
	
	Before the proof of the main result, let us state the following which gives a generating set as an $R$-module for a free code $C$ in $R[G]$. 
	
	\begin{prop} \label{two representations}
		Let $C$ be a free code in $R[G]$ with the following representation (cf. (\ref{disj unions})):
		$$C=\underset{1\leq i \leq t} {{\dot{\bigcup}}} (c_i+\gamma C).$$
		Let $S:=\{c_2,\ldots , c_t\}$. Then any element of $C$ can be represented as sum of the elements in 
		$$S \cup \gamma S \cup \cdots \cup \gamma^{v-1} S.$$
	\end{prop}
	
	\begin{proof}
		By Proposition \ref{NS-2}, we have
		$$C=\tilde{C}\cup \gamma \tilde{C} \cup \cdots \cup \gamma^{v-1} \tilde{C}\cup \{0\},$$
		where $\tilde{C}=C\setminus \gamma C$. Since cosets modulo $\gamma C$ partition $C$, and recalling that $c_1=0$, we have
		\begin{eqnarray*}
			\tilde{C} & = & (c_2 +\gamma C) \ \dot{\cup} \cdots \ \dot{\cup} \ (c_t+\gamma C),\\
			\gamma C & = & \gamma \tilde{C} \cup \cdots \cup \gamma^{v-1} \tilde{C} \cup \{0\}.
		\end{eqnarray*}
		Hence, 
		$$\tilde{C}=\underset{2\leq i \leq t} {{\dot{\bigcup}}} (c_i+\gamma C)= \underset{2\leq i \leq t} {{\dot{\bigcup}}} \left(c_i+(\gamma \tilde{C} \cup \cdots \cup \gamma^{v-1} \tilde{C} \cup \{0\})\right).$$
		Since $\gamma^v=0$, we have
		\[\begin{array}{lll}
		\gamma^{v-1}\tilde{C} & = & \displaystyle{\bigcup_{i=2}^t}\left\{ \gamma^{v-1}c_i\right\} , \\
		\gamma^{v-2}\tilde{C} & = & \displaystyle{\bigcup_{i=2}^t} \left(\gamma^{v-2}c_i +(\gamma^{v-1}\tilde{C})\right)\\
		& = &  \displaystyle{\bigcup_{i=2}^t} \left(\gamma^{v-2}c_i +\left(\displaystyle{\bigcup_{i=2}^t} \left\{\gamma^{v-1}c_i\right\} \right)\right) .
		\end{array}\]
		Continuing in the same manner until $\gamma \tilde{C}$, we obtain the desired result.
	\end{proof}
	
	We are ready to prove the main result for LCP of group codes (2-sided ideals) over a chain ring.
	
	\begin{thm}\label{main result-3}
		Let  $(C,D)$  be an LCP of group codes in $R[G]$, where $R$ is a finite chain ring and $G$ is a finite group. Then $C$ and $ D^{\bot} $ are equivalent codes. 
	\end{thm}
	
	\begin{proof}
		By Proposition \ref{duals}, $\varphi(C)$ and $\varphi(D^\bot)$ are equivalent codes. Let $\tau$ be the permutation between them (i.e. $\varphi(\tau(C))=\varphi(D^\bot)$). Note that $(C^\bot,D^\bot)$ is also an LCP of codes in $R[G]$ by Proposition \ref{dual LCP}, and hence $(\varphi(C^\bot),\varphi(D^\bot))$ is LCP in $\F_q[G]$ (Proposition \ref{lem2}). If $\{c'_1=0,c'_2,\ldots ,c'_s\}$ denotes the coset representatives of $C^\bot$ modulo $\gamma C^\bot$ and $\{d_1=0,d_2,\ldots ,d_t\}$, as before, denotes the coset representatives of $D^\bot$ modulo $\gamma D^\bot$, we have
		\begin{equation}\label{elements}
		\F_q[G]=\varphi(C^\bot)\oplus \varphi(D^\bot)=\{\varphi(c'_i)+\varphi(d_j): \ 1\leq i \leq s , \ 1\leq j \leq t \}.
		\end{equation}
		Since $C$ is free, $\tau(C)$ is also a free code in $R[G]$ and partitions as
		$$\tau(C)=\underset{1\leq i \leq t} {{\dot{\bigcup}}} (\tau(c_i)+\gamma \tau(C))   \ \ \mbox{(cf. (\ref{disj unions})),}$$
		where $\{c_1=0,c_2,\ldots ,c_t\}$ is the set of coset representatives of $C$ modulo $\gamma C$. 
		
		
		If $\tau(C)\cap C^\bot$ contains an element $x$ in a coset $c'_i +\gamma C^\bot$ for some $i\in \{2,\ldots , s\}$, then 
		$$\varphi(x)=\varphi (c'_i)\not\in \varphi (\tau(C))=\varphi(D^\bot)=\left\{\varphi (d_1)=0, \varphi (d_2), \ldots , \varphi(d_t) \right\} \ \mbox{(cf. (\ref{elements}))}.$$
		Therefore $\tau(C)\cap C^\bot$ is contained in $\gamma C^\bot$, hence in $\gamma \tau(C)$ (cf. Proposition \ref{NS-2} (iii)). Let $x\in \tau(C)\cap C^\bot$ be $x=\gamma \tau(c(1))=\gamma c'(1)$, where $c(1)\in C$ and $c'(1)\in C^\bot$. Then $\gamma (\tau(c(1))-c'(1))=0$ and hence the difference $\tau(c(1))-c'(1)$ is a multiple of $\gamma ^{v-1}$:
		$$\mbox{i.e. \ } \tau(c(1))=c'(1)+\gamma^{v-1}y_1, \ \mbox{for some $y_1\in R[G]$}.$$
		If $c'(1)\in C^\bot \setminus \gamma C^\bot$, then $\varphi (\tau(c(1)))=\varphi (c'(1))\not\in \varphi(\tau(C))$ again. Hence, $c'(1)=\gamma c'(2)$ for some $c'(2)\in C^\bot$ and 
		$$x=\gamma^2 c'(2)=\gamma^2 \tau(c(2)),$$
		where $c(2)\in C$. This yields $\gamma^2 (\tau(c(2))-c'(2))=0$ and hence the difference $\tau(c(2))-c'(2)$ is a multiple of $\gamma ^{v-2}$. In other words, $\tau(c(2))=c'(2)+\gamma^{v-2}y_2$ for some $y_2\in R[G]$. By the same reasoning, $c'(2)\in \gamma C^\bot$ and hence
		$$x=\gamma^3 \tau(c(3))=\gamma^3c'(3) \ \mbox{for some $c(3)\in C$ and $c'(3)\in C^\bot$}.$$ 
		Continuing in this manner, we conclude that the element $x$ in $\tau(C)\cap C^\bot$ must be $\{0\}$.
		
		Note that any permutation does not necessarily take an ideal of $R[G]$ to an ideal of $R[G]$. However $\tau$ does, as noted in Remark \ref{isom of group rings}, since it is induced from an automporhism of $G$. So, $\tau(C)$ is an ideal of $R[G]$.  By (\ref{perm assump-2}), we have (for all $1\leq i \leq t$)
		$$\tau(c_i)=d_i+\gamma x + \gamma y,$$
		for uniquely determined $x\in D^\bot$ and $y\in C^\perp$, since $R[G]=C^\perp \oplus D^\perp$. Let $1=a+b$ for $a\in C^\perp , b\in D^\perp$. Then, $\tau(c_i)= \tau(c_i)a + \tau(c_i)b$. Since $\tau(C)$ is an ideal, $\tau(c_i)a$ belongs to both $\tau(C)$ and $C^\perp$, whose intersection is $\{0\}$ (observe that we use the fact that $\tau(C)$ and $C^\bot$ are 2-sided ideals). Hence, 
		$$ \tau(c_i)=(d_i+\gamma x + \gamma y)b= (d_i+\gamma x)b + \gamma yb.$$
		Note that $yb=0$ since it belongs to $C^\perp \cap D^\perp=\{0\}$ (again, both codes are 2-sided ideals). Hence, $\tau(c_i) \in D^\perp$ for each $i$. This implies, by Proposition \ref{two representations}, that $\tau(C)\subset D^\perp$. Since $\tau(C)$ and $D^\perp$ have the same cardinalities (cf. Remark \ref{proof-obs1}), we have $\tau(C)=D^\perp$. This concludes the proof.
	\end{proof}

\begin{rem}\label{2-sided}
Note that Theorem \ref{main result-3} does not hold for LCP of one-sided group codes $(C,D)$ in $R[G]$. The proof of Theorem \ref{main result-3} utilizes the permutation $\tau$ between $\varphi(C)$ and $\varphi(D^\bot)$, whose existence is due to the result of Borello et al. and requires the codes in $\F_q[G]$ to be 2-sided group codes. Aside from this, 2-sided ideal structure of $C,C^\bot$ and $D^\bot$ are used, and pointed out, in several steps of the proof. 
\end{rem}

	\section{Acknowledgment}
	We would like to thank the reviewers for their valuable comments, which drastically improved the manuscript. The first author is supported by the T\"{U}B\.{I}TAK project 215E200, which is associated with the SECODE project in the scope of the CHIST-ERA Program. The second author is partially funded by the Spanish Research Agency (AEI) under grant PGC2018-096446-B-C21. The third author visited the Institute of Mathematics of University of Valladolid during February-March 2019. She thanks the Institute for their kind hospitality.


\begin{thebibliography}{90}
		
		\bibitem {BDGNN} S. Bhasin, J.-L. Danger, S. Guilley, Z. Najm and X. T. Ngo, ``Linear complementary dual code improvement to strengthen encoded circuit against hardware Trojan horses", \emph{IEEE International Symposium on Hardware Oriented Security and Trust (HOST)}, May 5-7, 2015.
		
	\bibitem{Rama}
		Sanjit Bhowmick and Alexandre Fotue-Tabue and Edgar Mart\'inez-Moro and Ramakrishna Bandi and Satya Bagchi. ``Do non-free LCD codes over finite commutative Frobenius rings exist?",  Des. Codes Cryptogr. 88, 825–840 (2020).
		
		
		
		\bibitem{BCW} M. Borello, J. de la Cruz, W. Willems, ``A note on linear complementary pairs of group codes", \emph{Discrete Math.}, vol. 343, 111905.
		
		
		\bibitem{BCCGM2014} J. Bringer, C. Carlet, H. Chabanne, S. Guilley, and H. Maghrebi, ``Orthogonal direct sum masking - a smartcard friendly computation paradigm in a code, with builtin protection against side-channel and fault attacks", in WISTP, Springer, Heraklion, 2014, 40-56.
		
		\bibitem{CGOOS} C. Carlet, C. G\"{u}neri, F. \"{O}zbudak, B. \"{O}zkaya and P. Sol\'{e}, ``On linear complementary pairs of codes", \emph{IEEE Trans. Inform. Theory}, vol. 64, 6583-6589, 2018.
		
		
		
		\bibitem{GOS} C. G\"{u}neri, B. \"{O}zkaya and S. Say\i c\i, ``On linear complementary pair of $n$D cyclic codes", \emph{IEEE Commun. Lett.}, vol. 22, 2404-2406, 2018.
		
		
		
		\bibitem{K} I. Kaplansky, ``Projective modules", \emph{Ann. of Math (2)}, vol. 68, 372-377, 1958.
		
		\bibitem{LL} X. Liu and H. Liu, ``LCD codes over finite chain rings", \emph{Finite Fields Appl.}, vol. 34, 1-19, 2015.
		
		\bibitem{LW} Z. Liu and J. Wang, ``Linear complementary dual codes over rings", \emph{Des. Codes Cryptogr.}, vol. 87, 3077-3086, 2019.
		
		
		
		
		
		
		\bibitem{M} J.L. Massey, ``Linear codes with complementary duals", \emph{Discrete Math.}, vol. 106/107, 337-342, 1992.
		
		
		\bibitem{NS1} G.H. Norton and A. Salagean, ``On the structure of linear and cyclic codes over a finite chain ring", \emph{Appl. Algebra Engrg. Comm. Comput.}, vol. 10, 489-506, 2000.
		
		\bibitem{NS2} G.H. Norton and A. Salagean, ``On the Hamming distance of linear codes over a finite chain ring", \emph{IEEE Trans. Inform. Theory}, vol. 46, 1060-1067, 2000.
		
		
	\end{thebibliography}
\end{document}